\documentclass{article}
\usepackage[utf8]{inputenc}

\pdfoutput=1
\usepackage{import}
\usepackage{amssymb}
\usepackage{amsmath}
\usepackage{amsthm}
\usepackage{xspace}
\usepackage{comment}
\usepackage{natbib}
\usepackage{url}
\usepackage{enumitem}
\usepackage{hyperref}
\usepackage{lmodern}
\usepackage{newtxtext}
\usepackage{xfrac} 
\usepackage[capitalise]{cleveref}
\usepackage{tikz}
\usetikzlibrary{positioning}

\usepackage{algorithm}
\usepackage{algorithmic}

\usepackage{xcolor}

\usetikzlibrary{shapes,arrows,positioning,fit,arrows.meta}

\newtheorem{theorem}{Theorem}[section]

\newtheorem{corollary}[theorem]{Corollary}
\newtheorem{lemma}[theorem]{Lemma}

\newtheorem{claim}[theorem]{Claim}
\newtheorem{definition}[theorem]{Definition}
\ifdefined\corollary\else
\newtheorem{corollary}[theorem]{Corollary}
\fi

\usepackage{enumitem}

\newtheorem{openq}{Open Question }
\theoremstyle{remark}

\usepackage{tcolorbox}
\renewenvironment{openq}{\refstepcounter{openq}\begin{tcolorbox}[colback=gray!7, colframe=black, boxsep=1pt, left=4pt, right=4pt, rounded corners]\textbf{Open Question \theopenq:}}{\end{tcolorbox}}

\newcommand{\findrepr}{\textsc{ConstructChain}\xspace}

\newcommand{\status}{\textsc{status}}

\newcommand{\NestedBased}{\textsc{NestedBasedRepresentation}\xspace}

\newcommand{\MGC}{\textsc{ModifiedGreedyCapture}}

\newcommand{\FR}{\textsc{FindRepresentative}}

\newcommand{\ChainBased}{\textsc{ChainBasedRepresentation}\xspace}

\newcommand{\union}{\cup}

\DeclareMathOperator*{\argmin}{arg\,min}

\def\min{\qopname\relax n{min}}
\def\max{\qopname\relax n{max}}
\def\argmin{\qopname\relax n{argmin}}

\def\Pr{\qopname\relax n{\mathbf{Pr}}}

\newcommand{\Rr}{\mathcal{R}}

\def\A{\mathcal{A}}

\def\G{\mathcal{G}}
\def\H{\mathcal{H}}

\def\P{\mathcal{P}}

\newcommand{\eat}[1]{}

\newenvironment{lp*}{\begin{equation*}  \begin{array}{lll}}{\end{array}\end{equation*}}

\usepackage[margin=1in]{geometry}
\title{\textbf{Temporal Panel Selection in Ongoing Citizens’ Assemblies}}

\author{ 
\textbf{Yusuf Hakan Kalaycı}  
 and    \textbf{Evi Micha}\\
  University of Southern California\\
  \texttt{\{kalayci,evi.micha\}@usc.edu}, 
}
\date{}
\begin{document}

\maketitle

Permanent citizens’ assemblies are ongoing deliberative bodies composed of randomly selected citizens, organized into panels that rotate over time. Unlike one-off panels, which  represent the population in a single snapshot, permanent assemblies enable shifting participation across multiple rounds. This structure offers a powerful framework for ensuring that different groups of individuals are represented over time across successive panels.  In particular, it allows smaller  groups of individuals that may not warrant representation in every individual panel to be  represented across a sequence of them. We formalize this temporal sortition framework by requiring proportional representation both within each individual panel and across the sequence of panels. 

Building on the work of Ebadian and Micha (2025), we consider a setting in which the population lies in a metric space, and the goal is to achieve both proportional representation, ensuring that every   group of citizens receives adequate representation, and individual fairness, ensuring that each individual has an equal probability of being selected. We extend the notion of representation to a temporal setting by requiring that every initial segment of the panel sequence, viewed as a cumulative whole, proportionally reflects the structure of the population. We present algorithms that provide varying guarantees of proportional representation, both within individual panels and across any sequence of panels, while also maintaining individual fairness over time.


\section{Introduction}
In recent years, \emph{citizens’ assemblies} have emerged as a promising model for community-driven governance~\cite{stone2011luck, vanreybrouck2016against}. The core idea is to randomly select a panel from the population, provide them with balanced information on an issue, allow time for deliberation, and then collect their recommendations, which are taken into account by the relevant authorities. This approach has been adopted in a variety of local contexts, with many city governments and local authorities implementing citizens' assemblies to address policy issues, as well as  in several prominent recent examples at the national level—such as in Ireland~\footnote{\url{https://www.citizensinformation.ie/en/government-in-ireland/irish-constitution-1/citizens-assembly/}}— and at the continental level through the European Union\footnote{\url{https://citizens.ec.europa.eu/european-citizens-panels_en}}.
At the heart of citizens' assemblies is \emph{sortition}: a randomized selection process creating a panel that reflects the general population's diverse perspectives, not just technical expertise. To support this, the process aims for representativeness across key features (e.g., geography, education level) while ensuring every individual has a meaningful selection chance. In practice, organizers often implement stratified sampling with quotas on individual features—for instance, reserving a fixed percentage of seats for a particular age group or requiring that at least a certain proportion of representatives have a specific occupation~\cite{FGGH+21}. However, this method fails to account for more complex combinations of features—for example, individuals who are both from a specific region and work in a particular occupation. On the other hand, forcing  representation across all such combinations is typically infeasible; with six features each taking just three values, there are already over 700 disjoint groups, far more than a typical panel in practice.

As a middle ground,  recent works~\cite{SKMPS22, Ebadian2025} propose leveraging an underlying representation metric space for defining when a panel represents an underlying population in a more rigorous way. This space measures how well an individual represents another individual or group, with smaller distances indicating stronger representation; this space can be constructed based on features relevant to the application. Given such a metric space, \citet{Ebadian2025} adapt a notion of proportional representation previously used in multi-winner elections~\cite{lackner2022approval} and clustering~\cite{aziz2017justified, chen2019proportionally, micha2020proportionally} to define when a panel is proportionally representative of an underlying population. At a high level, the notion requires that, given a population of size \( n \) and a panel of size \( k \), any subset of the population of size at least \( q \cdot \sfrac{n}{k} \) is entitled to up to \( q \) representatives.

As citizens’ assemblies gain increasing acceptance, the next step in this progression is the growing establishment of \emph{permanent citizens’ assemblies}, an institutional innovation that enables ongoing participation rather than one-off deliberation. Unlike traditional, ad-hoc citizens’ assemblies that are convened for specific issues and disbanded afterward, these permanent bodies provide structured, recurring opportunities for citizens to participate directly in policymaking over extended periods. While the institution itself is continuous, its members are periodically rotated (typically every six months or annually) through sortition. The Ostbelgien Model in the German-speaking community of Belgium is one of the most prominent examples~\cite{niessen2019designing}.
Established in 2019, it includes a standing Citizens' Council that   provides citizens a permanent voice in the process of decision making  by rotating Citizens' Panels.  European Union has also  recently experimented with a similar format~\cite{alemanno2022permanent}. This new innovation 
 motivates the main  question in this work:
\begin{quote}
\emph{In ongoing citizens’ assemblies, can we ensure that smaller groups receive representation over time, even if they do not qualify in every individual panel, while still guaranteeing that every sufficiently large group is proportionally represented in each panel?}
\end{quote}

\paragraph{Temporal Sortition.}
Permanent citizens' assemblies create new opportunities to achieve broader representation across groups of individuals who share similar features. Smaller groups that may not qualify for representation on every individual panel can now be  represented across a sequence of panels.
To illustrate this, consider a toy example in which the population consists of four disjoint groups: the first makes up half of the population, the second a quarter, and the remaining two one-eighth each. Suppose we are selecting a panel of size 4. Proportionality would assign 2 seats to the first group, 1 to the second, and the last to either the third or fourth group, but not both.
However, if we construct two consecutive panels, we can alternate this final seat—assigning it to the third group in the first panel and to the fourth in the second. In this way, both smaller groups receive representation over time, even if not in every individual panel.

This simple example demonstrates how permanent citizens' assemblies can extend proportional representation beyond the limits of any single panel. To capture this formally, we introduce the \emph{temporal sortition framework}, which models how representation can be extended across a sequence of panels.
At a high level, the goal is to construct a sequence of randomly selected panels such that each panel proportionally represents every group that is large enough within that panel's population; groups that are not eligible in a single panel still receive representation once enough seats accumulate across the sequence; and each individual has an equal probability of being selected throughout the process.
Formally, given a population $V$ of size $n$ in a representation metric space, our objective is to construct a panel sequence \( P_1, P_2, \ldots, P_{\ell} \), where panel \( P_i \) has size \( k_i \). Leveraging notions of proportional representation similar to those utilized by~\citet{Ebadian2025} for static panels, and assuming the existence of such a metric space, our goal is to define a distribution over these \( \ell \) consecutive panels that satisfies the following properties:

\begin{enumerate}
    \item \textit{Representation per Panel:} Every subset of the population of size at least \( \sfrac{n}{k_i} \) (or \( q \cdot \sfrac{n}{k_i} \)) receives at least 1 (or \( q \)) representatives in panel \( P_i \).
    
    \item \textit{Representation Over Time:} 
    For every \( t \in [\ell] \), every subset of the population of size at least \( \sfrac{n}{\sum_{i=1}^t k_i} \) (or \( q \cdot \sfrac{n}{\sum_{i=1}^t k_i} \)) receives at least 1 (or \( q \)) representatives among the first \( t \) panels, i.e., in \( \cup_{i=1}^t P_i \).
    
    \item \textit{Individual Fairness:} Each individual is included in the union of all panels \( \cup_{i=1}^{\ell} P_i \) with equal probability, equal to \( \sfrac{\sum_{i=1}^{\ell} k_i}{n} \).

\end{enumerate}

 To define representation rigorously in the presence of a metric space, we draw on two established notions from the literature: \(\alpha\)\nobreakdash-Proportionally Fair Clustering (\(\alpha\)\nobreakdash-PFC) by~\citet{chen2019proportionally}, and \(\beta\)\nobreakdash-Proportionally Representative Fairness (\(\beta\)\nobreakdash-PRF)  by~\citet{aziz-lee-chu-vollen-proportional-clustering}. Roughly speaking, \(\alpha\)\nobreakdash-PFC guarantees that any group of at least \( \sfrac{n}{k} \) individuals in the metric space has 
at least  one representative whose distance from  some member of the group is  at most  $\alpha$
 times the group’s diameter.
 In contrast, \(\beta\)\nobreakdash-PRF ensures that any group of size at least \( q \cdot \sfrac{n}{k} \) has at least \(q\) representatives 
  whose distance from  some member of the group is  at most  $\beta$
 times the group’s diameter.

\paragraph{Technical Challenge and Our Contributions.}

In~\Cref{sec:federated-sortition}, as a warm-up, we consider the simpler setting where the objective is to ensure proportional representation for each individual panel as well as for the overall collection of panels, which we refer to as the \emph{global panel} \( \bigcup_{i \in [\ell]} P_i \). We begin with the case of creating \(\ell\) panels of size \(k\). A natural approach is to proceed in two steps. First, we select \(\ell \cdot k\) individuals from the population in a way that provides a constant-factor approximation to PRF, using known algorithms from the literature. Next, we partition this set of \(\ell \cdot k\) individuals into \(\ell\) panels of size \(k\), again applying known algorithms   so that each panel fairly represents this smaller group. Together, these two steps ensure that both the individual panels and the global panel satisfy a constant-factor approximation to PRF.  More generally, we  investigate what happens  when a smaller panel is constructed from a larger one, i.e., we first form a large panel \(P_1\) of size \(k_1\) that satisfies \(\alpha\)-PRF with respect to the whole population, and then extract a smaller panel \(P_2\)  of size \(k_2\) from $P_1$ that satisfies \(\alpha\)-PRF with respect to  \(P_1\).  We show that  if \(k_1\) is a multiple of \(k_2\), then \(P_2\) also represents the original population within a \((2\alpha\beta + \beta)\)-approximation to PRF. Surprisingly, however, when \(k_1\) is \emph{not} divisible by \(k_2\), there exist instances where the smaller panel \(P_2\) fails to satisfy \emph{any} meaningful approximation of proportional representation for the population. This highlights that the problem is technically far more subtle than it might appear at first glance.

In~\Cref{sec:prefix-and-panel}, we strengthen the requirements to demand representation for each panel, the global panel, and every prefix \( P_{\leq t} = \bigcup_{i \in [t]} P_i \). This significantly increases the difficulty of the problem, since the goal is to construct a representative global panel, decompose it into representative individual panels, and order them so that every prefix $P_{\leq t}$ is also representative, as if the task were to build a panel of size equal to the total size of that prefix, all while also ensuring individual fairness. Though this appears challenging, we design an algorithm for $\ell$ panels of size $k$ that guarantees individual fairness and a PFC approximation, with respect to both each individual panel and each prefix $P_{\leq t}$, that scales exponentially with the number of panels. The main property that the algorithm exploits is that under proportionality, any subset of individuals that deserves representation in smaller panels should also be represented in larger panels. The algorithm thus constructs groups for each prefix size, enforcing that groups for larger panels are nested within those for smaller panels. The algorithm then carefully leverages this nested structure to guarantee proportional representation for each panel and for every prefix. This hierarchical approach, however, comes at the cost of an exponential blow-up in the approximation factor with respect to the number of panels  $\ell$.

Finally, in~\Cref{sec:prefix} we consider a relaxed but still highly demanding requirement. Instead of enforcing proportional representation for each individual panel we require it only for every prefix \( P_{\leq t} = \bigcup_{i \in [t]} P_i \) with \( t \leq \ell \). While this abandons per-panel guarantees, it still captures a strong notion of representation by ensuring that cumulative representation is preserved over time. We show that this relaxation makes the problem  more tractable and we present an algorithm that achieves a constant-factor approximation of PFC with respect to every prefix, while still guaranteeing that each individual is included in the global panel with equal probability. The algorithm is particularly intriguing and the key idea  is to replace the  hierarchical structure with a more flexible construction that carefully links together the groups that should be represented at different panel sizes, that is, across prefixes. This design ensures that representation achieved for earlier prefixes automatically extends to later ones, preventing error from accumulating and ultimately yielding a constant-factor bound.

Taken together, these results show that temporal sortition algorithms can achieve both per-panel and cumulative (over-time) representation guarantees, along with individual fairness. At the same time, our findings uncover subtle structural challenges that make the problem far from straightforward. Throughout the paper we highlight several intriguing open questions that we hope will inspire further research.

\paragraph{Related Work.}

The design of citizens’ assemblies that are representative of the broader population has received significant attention in the computer science literature~\cite{FGGH+21, FGGP20,flanigan2024manipulation,baharav2024fair, Ebadian2025,SKMPS22, CaragiannisMichaPeters2024}. However, to the best of our knowledge, no prior work has addressed the problem of designing a sequence of panels that are representative of the population both individually and collectively over time. The only conceptually related work is the recent proposal by \citet{halpern2025federated}, who introduce federated assemblies: a hierarchical model in which assemblies are connected through a directed acyclic graph, and members of higher-level assemblies are drawn from lower-level ones. Their algorithms ensure individual fairness (each person has an equal probability of selection), ex ante fairness (each subassembly is expected to receive representation proportional to its size), and ex post fairness (the realized allocation closely approximates the expected proportions). In contrast, our work focuses on a sequence of stand-alone panels that, taken together over time, provide comprehensive representation, shifting the focus from structural to longitudinal representational equity.

Metric proportional representation has emerged as a central theme across computational social choice, clustering, and data summarization. Early clustering research introduced proportionally fair clustering \citep{chen2019proportionally} and individual fairness \citep{jung-individual-fairness} establishing the idea of ensuring representation for sufficiently large point sets. These concepts subsequently influenced committee selection, spawning frameworks for proportionally representative committees \citep{kalayci2024proportional} and proportionally representative fairness \citep{aziz-lee-chu-vollen-proportional-clustering}.

In sortition, \citet{SKMPS22} first initiated the idea of utilizing a  metric space for achieving representation, but they focus on selecting a panel that maximizes the social welfare. \citet{Ebadian2025} later introduced the idea of the  proportional representation over a metric space and designed the Fair Greedy Capture algorithm that maintains individual fairness while ensuring constant-factor approximation to the ex post core. These single-shot formulations ($\ell = 1$) serve as foundation for our temporal framework extending to sequential multi-panel settings. We refer readers to \citep{kellerhals:peters:proportional-fairness} for the landscape of metric proportional representation and their relationships.

Fairness across time horizons has gained momentum in perpetual voting \citep{Lackner:Perpetual-Voting, Lackner:Martin:Maly:Proportional-Decision-Perpetual, Bulteau:Hazon:Page:Rosenfeld:Talmon:JR-for-perpeutal}, temporal committee selection \citep{elkind:obraztsova:teh:temporal-fairness-multiwinner, Elkind:Obraztsova:Peters:Teh:verifying-proportionality-temporal-voting, do:hervouin:lang:skowron:online-voting}, temporal clustering \citep{dey:rossi:Sidiropoulos:temporal-clustering, dey:rossi:sidiropoulos:temporal-hierarchical-clustering}, repeated matching \citep{gollapudi-matching, caragiannis-matching, trabelsi-matching}, and sequential allocation \citep{bampis:fair-allocation, igarashi:fair-allocation}. Motivated by these works and the need for time-robust citizens' assemblies, we extend single-shot metric sortition to multi-step settings where groups underrepresented in individual panels achieve fair representation across time.

\section{Preliminaries}
\label{sec:preliminaries}

For \( m \in \mathbb{N} \), let \( [m] = \{1, \ldots, m\} \). We denote the population by \(V= [n] \).  The individuals are embedded in an underlying \emph{representation metric space} equipped with a distance function \( d \), where the distance between any two individuals \( i \) and \( j \) is denoted by \( d(i, j) \). We assume that \( d \) is a pseudo-metric 
\footnote{A \emph{metric} is a non-negative function $d$ on pairs satisfying: (i) $d(x,x) = 0$ for all $x$, (ii) \emph{symmetry}: $d(x,y) = d(y,x)$ for all $x, y$, (iii) \emph{triangle inequality}: $d(x,z) \leq d(x,y) + d(y,z)$ for all $x, y, z$, and (iv) \emph{positivity}: $d(x,y) > 0$ whenever $x \neq y$. A \emph{pseudo-metric} relaxes the positivity requirement, allowing distinct points to have zero distance.}. 
An instance of our problem is fully specified by the set of individuals and their pairwise distances. For simplicity, we use $d$ to refer to both the instance and the underlying distance function.
For any individual $v \in V$ and radius $r \geq 0$, we define the \emph{ball} $B(v,r) = \{u \in V \mid d(v,u) \leq r\}$, i.e.  the set of all individuals within distance 
$r$ in the metric space from a given individual 
$v$.

A \emph{temporal selection algorithm}  $\A$ takes as input the population, the metric, the number of panels $\ell$ and panel sizes $k_1, k_2 \dots k_\ell$ and outputs a probability distribution  over a sequence of $\ell$ disjoint panels $P_1, \ldots P_\ell$ where $P_i$ has size equal to $k_i$.  A sequence of panels, $P_1, \ldots P_\ell$, is in the support of $\A$, if the algorithm returns this sequence with positive probability. Throughout this work, we assume that $\sum_{t=1}^{\ell} k_t \leq n$, so that the total number of panel seats does not    exceed the population size. This is the most natural setting in practice, as citizens' assemblies typically involve a few hundred members per panel, and even over decades of operation, the cumulative number of participants remains orders of magnitude smaller than the relevant population.

\paragraph{Representation Axioms.}
We start by defining two notions of proportional representation that have been proposed in the literature for individual panels. 


    \begin{definition}[$\alpha$-Proportionally Fair Clustering ($\alpha$-PFC)\cite{chen2019proportionally}]
    \label{def:pfc}
    A panel $P \subseteq V$ of size $k$ is $\alpha$-proportionally fair for $\alpha \geq 1$ if, for any subset $S \subseteq V$ of size at least $\sfrac{n}{k}$, there exists an individual $v \in S$ and a panel member $p \in P$ such that
    $$d(v, p) \leq \alpha \cdot \min_{y \in V} \max_{u \in S} d(u,y).$$ 
    \end{definition}
In words, proportionally fair clustering ensures that no group of individuals of size at least $\sfrac{n}{k}$ can identify an alternative representative not in the current panel that all its members would strictly prefer over their closest representative in the selected panel.

    \begin{definition}[$\beta$-Proportionally Representative Fairness ($\beta$-PRF)~\cite{aziz-lee-chu-vollen-proportional-clustering}]
    \label{def:prf}
    A panel $P \subseteq V$ of size $k$ is $\beta$-proportionally representative  if for any set of individuals $S \subseteq V$ of size at least $q \cdot \sfrac{n}{k}$ where the maximum pairwise distance within $S$ is $r$, we have:
    $$ \left| P \cap \bigcup_{v \in S} B(v, \beta \cdot r) \right| \geq q. $$
    \end{definition}
In words, proportionally representative fairness ensures that for every group of individuals of size at least $q \cdot \sfrac{n}{k}$ with maximum pairwise distance $r$, there exist at least $q$ representatives in the panel, each within distance $r$ of some individual in the group.

    Intuitively, proportionally fair clustering ensures that coalitions of size at least $\sfrac{n}{k}$ receive representation by at least one panel member, while proportionally representative fairness requires that larger coalitions (of size at least $q \cdot \sfrac{n}{k}$) receive proportionally many representatives (specifically, $q$ representatives). Recent work has established theoretical relationships between these concepts: \citet{aziz-lee-chu-vollen-proportional-clustering} and  \citet{kellerhals:peters:proportional-fairness} showed that $1$-PRF implies $(1+\sqrt{2})$-PFC, and 
    this can be generalized to demonstrate that $\beta$-PRF implies $(1+\sqrt{2}) \cdot \beta$-PFC. 

\paragraph{Axioms for Temporal Sortition.} 
We evaluate the fairness and representation guarantees of a distribution  across a series of panels $  {P_1, \dots, P_\ell}$ using the following axioms 
:
    
\begin{enumerate}
  \item \textbf{Individual Fairness:} 
  Each individual should have equal selection probability across the entire process. We say that a selection algorithm satisfies \emph{individual fairness} if for all $v \in V$,
  $\Pr_{P_1,\ldots, P_{\ell} \sim \A}\left[v \in P\right] = \sfrac{\sum_{i=1}^\ell k_i}{n},$
  where $P = \bigcup_{i=1}^\ell P_i$. 

  \item \textbf{Individual Panel Representation:} 
 For any representation axiom $\Pi$, we say that a selection algorithm  $\A$ satisfies $\Pi$ at the ``panel level'' if for every $\P=\{P_1, \dots, P_\ell\}$ in the support of $\A$ every $P_i$ satisfies $\Pi$ for population $V$ and panel size $k_i$, simultaneously. In this work, we focus on $\alpha$-PFC and $\beta$-PRF definitions.
  
  \item \textbf{Global Panel Representation:} The union of all panels should satisfy  representation axioms, providing long-term representation guarantees. Given a representation axiom $\Pi$, we say that a selection algorithm $\A$  satisfies \emph{$\Pi$ at the global level}  if  for every $\P=\{P_1, \dots, P_\ell\}$ in the support of $\A$, $\bigcup_{i=1}^\ell P_i$ satisfies axiom $\Pi$ for population $V$ and panel size $\sum_{i=1}^\ell k_i$.

  \item \textbf{Prefix Representation:} For more stringent requirements beyond global panel representation, we  demand that for every  $t \in [\ell]$, the cumulative selection panel  of the first $t$ panels  maintains proportional representation, ensuring representation quality at every stage.  Given a proportionality axiom $\Pi$, we say that that a selection algorithm  $\A$ satisfies $\Pi$ at the ``prefix level''  if  for every $\P=\{P_1, \dots, P_\ell\}$ in the support of $\A$,    the cumulative panel $P_{\leq t}$ for each $t \in [\ell]$ satisfies representation axiom $\Pi$ for population $V$ with panel size $\sum_{i=1}^t k_i$.
  
\end{enumerate}

\paragraph{Paper Organization.}
In~\Cref{sec:federated-sortition}, we consider the warm-up setting where the goal is to ensure $O(1)$-PRF for each individual panel and the global panel. In~\Cref{sec:prefix-and-panel}, we strengthen the requirements to demand both panel-level and prefix-level representation simultaneously, achieving $O(4^\ell)$-PFC guarantees. Finally, in~\Cref{sec:prefix}, we relax the per-panel requirement and focus solely on prefix-level representation, recovering a constant-factor $O(1)$-PFC approximation for every prefix. All three algorithms guarantee individual fairness.

\section{Warm-up: Proportional Representation Per Panel and Global Panel}
\label{sec:federated-sortition}

As a warm-up, we first consider the case where the goal is to achieve representation with respect to each individual panel  and the global panel, without prefix representation concerns. This setup can be interpreted as a federated sortition problem where $\ell$ panels operate simultaneously, and we seek to ensure representation both within each individual panel and across the union of all panels~\cite{halpern2025federated}.

To design algorithms for this case, we begin by examining how representation guarantees behave when rules satisfying PRF axioms are applied in sequence. In particular, consider a population $V$ and suppose we apply a selection algorithm satisfying the $\alpha$-PRF axiom to obtain a panel $P_1$ of size $k_1$. We then apply another selection algorithm satisfying the $\beta$-PRF axiom  to the selected panel $P_1$, yielding a smaller panel $P_2 \subseteq P_1$ of size $k_2$. This raises a natural question: what proportionality guarantees does $P_2$ achieve with respect to the original population $V$ and panel size $k_2$?  In the next theorem, we show that this approach provides an approximation guarantee that depends on $\alpha$ and $\beta$ when  $k_1$ is divisible by $k_2$. But surprisingly, when $k_1$ is not divisible by $k_2$, then $P_2$ may not provide any finite approximation with respect to $V$.

\begin{theorem} \label{thm:nested-representation}
    Given a population $V$ and panel sizes $k_1$ and $k_2$, let $P_1$ be a panel of size $k_1$ that satisfies $\alpha$-PRF for population $V$, and let $P_2$ be a smaller panel of size $k_2$ that satisfies $\beta$-PRF for population $P_1$. Then,
    \begin{itemize}
        \item when $k_1$ is divisible by $k_2$, $P_2$ satisfies $(2\alpha\cdot\beta+\beta)$-PRF for population $V$.
        \item there exist panel sizes $k_1$ and $k_2$ where $k_1$ is not divisible by $k_2$ and an instance in which $P_2$ does not satisfy $\alpha$-PRF (or $\alpha$-PFC) for population $V$, for any finite $\alpha$.
    \end{itemize}

\end{theorem}

\begin{proof}
    We start by the case where $k_1$ is divisible by $k_2$.
    Consider an arbitrary coalition $S \subseteq V$ of size at least $q \cdot \frac{n}{k_2} = q \cdot \frac{n}{k_1} \cdot \frac{k_1}{k_2}$ for some integer $q \geq 1$ and let $r=\max\{d(u,v): u,v\in S\}$ be the diameter.

    Since $P_1$ satisfies $\alpha$-PRF with respect to $V$ and panel size $k_1$, we have
    $$\left| P_1 \cap \bigcup_{v \in S} B(v, \alpha \cdot r) \right| \geq q \cdot \frac{k_1}{k_2}.$$ 
    
    Let $T = P_1 \cap \bigcup_{v \in S} B(v, \alpha \cdot r)$. To bound the diameter of $T$, consider any two individuals $u, v \in T$. By definition of $T$, there exist individuals $u', v' \in S$ such that $u \in B(u', \alpha \cdot r)$ and $v \in B(v', \alpha \cdot r)$. Applying the triangle inequality and using $d(u', v') \leq r$, we obtain:
    $$d(u, v) \leq d(u, u') + d(u', v') + d(v', v) \leq \alpha \cdot r + r + \alpha \cdot r = (2\alpha + 1) \cdot r.$$
    
    Thus, the diameter of $T$ is at most $(2\alpha + 1) \cdot r$. Since $P_2$ satisfies $\beta$-PRF with respect to population $P_1$, we have
    $$ \left| P_2 \cap \bigcup_{v \in T} B(v, \beta \cdot (2\alpha + 1) \cdot r) \right| \geq q, $$
    which establishes that $P_2$ satisfies $(2\alpha\cdot\beta + \beta)$-PRF for population $V$.

   For the case, where  $k_1$ is not divisible by $k_2$ consider a population $V$ of size $9$ on the real axis with locations $V = \{x_1=0, x_2=0, x_3=1, x_4=10, x_5=10, x_6=10, x_7=11, x_8=11, x_9=11\}$. Let $k_1=4$ and $k_2=3$.

    The panel $P_1 = \{x_1, x_3, x_4, x_7\}$ satisfies $1$-PRF for population $V$. The panel $P_2 = \{x_1, x_3, x_4\}$ satisfies $1$-PRF for population $P_1$. However, $P_2$ does not satisfy any approximate PRF for $V$ since coalition $\{x_7, x_8, x_9\}$ requires a representative at location $11$, which $P_2$ lacks.

\end{proof}

Based on the above theorem, for the simple case where all panels have the same  size $k$ across all $\ell$ panels, we can utilize  algorithms from the literature for achieving  approximately PRF for each panel and  the global panel and individual fairness properties. In particular, we can first apply the algorithm called, Fair Greedy Capture by~\citet{Ebadian2025}, for choosing a panel $P$ with size $\ell \cdot k$ which satisfies $6$-PRF  and ensures that each individual is selected with the same probability~\footnote{Technically,~\citet{Ebadian2025} establish this approximation under a slightly different notion of representation; however, by adapting their arguments, the guarantee carries over to PRF as well.}. Then we can apply, the algorithm called, Metric Expanding Approval by~\citet{aziz-lee-chu-vollen-proportional-clustering} with parameters $P$ and $k$  for partitioning the $\ell\cdot k$ individuals in $P$ into $k$  groups of size $\ell$ each.  At a high level, when panel size evenly divides population size, Expanding Approval Rule proceeds by simultaneously expanding balls around each representative in 
$P$. Once a ball captures $\ell$ individuals, these are grouped together, and the process continues with the remaining representatives until $k$ such groups are formed.  \citet{aziz-lee-chu-vollen-proportional-clustering} show that when Metric Expanding Approval  is applied to an underlying population, selecting one representative from each resulting group yields a panel that satisfies $2$-PRF with respect to the population. Building on this, we assign to each individual panel \( P_i \) one representative drawn uniformly at random from each of the \( k \) groups. This guarantees that every panel receives exactly one representative from each group, while each representative is assigned to a panel with probability \( \sfrac{1}{\ell} \).
 Combined with the fact that each individual in the population is selected into \( P \) with  probability  \( \sfrac{\ell\cdot k}{n} \), it follows that every individual is assigned into a panel \( P_i \) with probability \( \sfrac{k}{n} \). By ~\Cref{thm:nested-representation}, we then obtain the following corollary.

\begin{corollary}
    Given population $V$,  there exists a polynomial time algorithm that returns    $\ell$ panels, $P_1, \dots P_\ell$ of size  $k$ each, such that  the global panel, i.e. $\cup_{i\in [\ell]} P_i$, is $6$-PRF, each panel $P_i$ is  $26$-PRF, and each $v\in V$ is included in panel $P_i$ with probability $\sfrac{k}{n}$.
\end{corollary}

However, despite significant effort, we were not able to generalize this result for the case where the panels have different sizes, and the following question remains open.

\begin{openq}
    Does there exist a  distribution over a sequence of panels, $P_1, \ldots, P_{\ell}$, where $P_i$ has size $k_i$ such that each individual  panel $P_i$  satisfies $O(1)$-PRF  and the global panel $\cup_{i\in [\ell]}P_{i}$ satisfies $O(1)$-PRF?
\end{openq}

\section{Prefix and Panel Level Representation}\label{sec:prefix-and-panel}

In this section, we turn our attention to the setting where the goal is to ensure not only global representation but also \emph{prefix representation}. That is, for every time step \( t \in [\ell] \), we require that the cumulative panel \( P_{\leq t} = \bigcup_{j=1}^t P_j \) provides appropriate representation. Ideally, one would show that there exists a distribution over a sequence of panels, $P_1, \ldots, P_{\ell}$, where $P_j$ has size $k$ (or more generally $k_j$), such that each individual panel $P_j$ and each prefix panel $P_{\leq t}$ satisfy a constant-factor approximation of PRF or at least PFC. Despite considerable effort, we are unable to prove whether such a distribution always exists, leaving a tantalizing  question open.

\begin{openq}
    Does there exist a  distribution over a sequence of panels, $P_1, \ldots, P_{\ell}$ such that each individual  panel $P_j$  and each prefix panel $P_{\leq t}$  satisfies $O(1)$-PRF or $O(1)$-PFC? 
\end{openq}

Instead, here we present a novel algorithm that provides weaker approximation guarantees with respect to individual and prefix representation for the case where each panel has size $k$. In particular, we show that there exists a polynomial-time algorithm that achieves an approximation to PFC within each individual panel and each prefix panel, which  grows exponentially with the number of panels $\ell$, while also ensuring that each individual is selected to participate in one of the panels  with equal probability.

\begin{theorem}\label{thm:exponential}
    There exists a polynomial time algorithm that returns a distribution over a sequence of $\ell$ panels $P_1, \ldots, P_{\ell}$ of size $k$ each such that:
    \begin{itemize}
        \item  For each $t\in [\ell]$, panel $P_t$ satisfies $O(4^{\ell})$-PFC;
         \item For each $t\in [\ell]$,  prefix panel $P_{\leq t}$  satisfies $O(4^{\ell-t})$-PFC;
         \item Each individual is selected in $\cup_{t\in [\ell]} P_{t}$ with probability $\sfrac{\ell\cdot k}{n}$. 
    \end{itemize}
     
\end{theorem}

At a high level, our main algorithm, \NestedBased, operates in two phases. In the first phase, it constructs collections $\G^t$ of disjoint groups of individuals, each of size at least $\sfrac{n}{t \cdot k}$ for every $t \in [\ell]$. Each group in $\G^t$ represents a subset of the population that should be represented within the first $t$ panels.
The algorithm proceeds in reverse order, starting from \( t = \ell \) and forming the collection of groups $\G^t$. As \( t \) decreases, it continues forming new groups while ensuring that previously constructed groups are preserved and hierarchically nested within the new groupings.

In the second phase, the algorithm traverses this hierarchy (visualized as a tree) starting from groups corresponding to smaller values of 
$t$ (i.e., larger groups) and proceeding toward groups corresponding to  larger values of
$t$ (i.e., smaller groups) that are nested within the previous ones. 
 For each such path, the algorithm identifies the terminal group, i.e. a group corresponding to the larger value of $t$ that does not have any  other group nested and selects one individual from this terminal group as the path's representative.  The algorithm then assigns these selected representatives to specific panels in sequential order, ensuring that the assignment respects both  per-panel representation  and  prefix-level representation.
 
Below, we describe these two main phases separately in detail.

\textbf{Tree Construction.} 
First, we describe an algorithm called, \MGC\ which is a variation of \citet{chen2019proportionally}'s algorithm that takes population $V$, metric $d$, target panel size parameter $K$, and a partition  $\G$ of the population into disjoint groups. The algorithm expands balls around each individual, initially marked incomplete. A ball captures a group in $\G$ only when it captures all its members; partial capture does not count. When an incomplete ball captures at least \( \sfrac{n}{K} \) individuals, it becomes \emph{complete}, its captured groups are consolidated as a single group and added to $\G'$, then disregarded from further processing. The algorithm continues expanding all  balls, disregarding any groups in $\G$ that are captured by complete balls. New groups are formed when incomplete balls capture total of at least $\sfrac{n}{K}$ uncovered individuals. The process terminates when all groups are disregarded, returning $\G'$.  By construction, each input group $G \in \G$ is either entirely contained within some output group $G' \in \G'$ or entirely excluded. For each group \( G \) created during the execution of \MGC, we denote by \( c_G \) its 
center, i.e., the point around which the ball was grown when the group was formed, 
and by \( r_G \) its radius, defined as the radius of the ball at the moment \( G \) 
was created. 
For complete algorithmic details, see \Cref{sec:mgc}.

\NestedBased starts by constructing a tree rooted at $\Rr$, initially containing singleton groups $\{\{v\}: v \in V\}$ as children. 
Then, from \( t = \ell \) down to \( 1 \), \NestedBased\ calls \MGC, which at iteration \( t \) takes \( t \cdot k \) as parameter \( K \) and uses the current children of \( \Rr \) as its group partition, outputting a collection \( \G^t \) in which each group contains at least \( \sfrac{n}{t \cdot k} \) individuals.
 After each execution, \NestedBased updates the tree structure as follows. For each group $H \in \G^t$, all groups contained in $H$ are removed from $\Rr$'s children, and $H$ itself becomes a new child of $\Rr$. 
 An example is illustrated at~\Cref{fig:combined}. 
 This tree structure ensures that if leaf node $v$ is selected (also referring to an individual), any group on the path from $\Rr$ to $v$ is approximately represented, as shown in the following lemma.

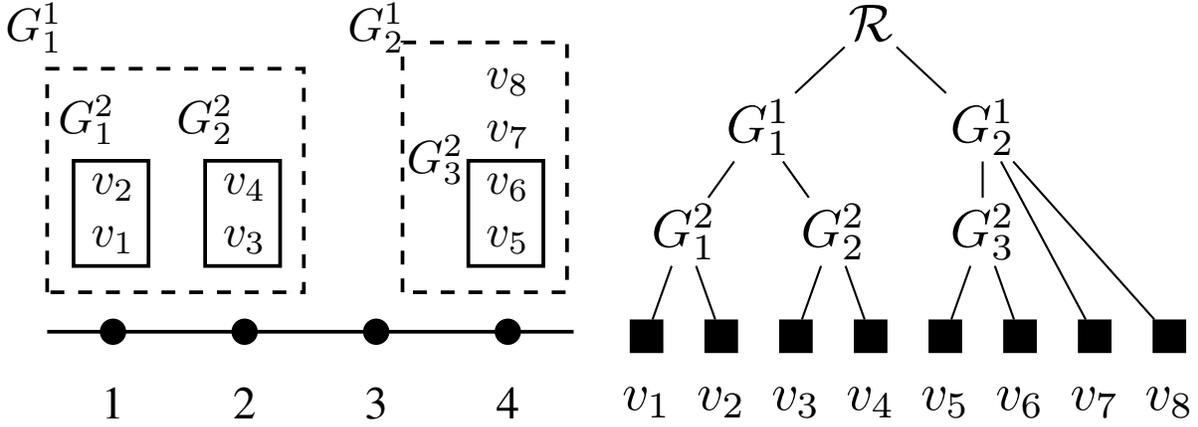
\begin{figure}[t]
    \centering
    \begin{minipage}{0.49\columnwidth}
        \centering
        \resizebox{\textwidth}{!}{

\begin{tikzpicture}[scale=1]

\draw[thick] (0.5,0) -- (4.5,0);

\foreach \x in {1,2,3,4} {
    \fill (\x,0) circle (0.1);
    \node[below] at (\x,-0.3) {\x};
}

\draw[thick] (0.7,0.5) rectangle (1.27,1.3);
\node at (1,0.7) {$v_1$};
\node at (1,1.1) {$v_2$};

\draw[thick] (1.7,0.5) rectangle (2.27,1.3);
\node at (2,0.7) {$v_3$};
\node at (2,1.1) {$v_4$};

\draw[thick] (3.7,0.5) rectangle (4.27,1.3);
\node at (4,0.7) {$v_5$};
\node at (4,1.1) {$v_6$};

\node at (4,1.5) {$v_7$};
\node at (4,1.9) {$v_8$};

\draw[dashed, thick] (0.5,0.3) rectangle (2.45,2);
\node[above] at (0.4, 2) {$G^1_1$};
\node[above] at (0.8, 1.3) {$G^2_1$};
\node[above] at (1.7, 1.3) {$G^2_2$};

\draw[dashed, thick] (3.2,0.3) rectangle (4.45,2.2);
\node[above] at (3,2) {$G^1_2$};
\node[above] at (3.45,1) {$G^2_3$};


\end{tikzpicture}
        }
    \end{minipage}
    \hfill
    \begin{minipage}{0.49\columnwidth}
        \centering
        \resizebox{\textwidth}{!}{

\begin{tikzpicture}[
  level distance=0.7cm,
  level 1/.style={sibling distance=1.5cm},
  level 2/.style={sibling distance=1cm},
  level 3/.style={sibling distance=0.5cm},
  every node/.style={circle, draw, fill=black, inner sep=1pt, minimum size=4pt},
  label/.style={rectangle, draw=none, fill=none, text=black}
]

\node[label] (root) {$\mathcal{R}$}
  child {
    node[label] {$G_1^1$}
    child {
      node[label] {$G_1^2$}
      child {
        node[rectangle, draw, fill=black, minimum size=6pt] {}
        node[label, below=0.3cm] {$v_1$}
      }
      child {
        node[rectangle, draw, fill=black, minimum size=6pt] {}
        node[label, below=0.3cm] {$v_2$}
      }
    }
    child {
      node[label] {$G_2^2$}
      child {
        node[rectangle, draw, fill=black, minimum size=6pt] {}
        node[label, below=0.3cm] {$v_3$}
      }
      child {
        node[rectangle, draw, fill=black, minimum size=6pt] {}
        node[label, below=0.3cm] {$v_4$}
      }
    }
  }
  child {
    node[label] (g21) {$G_2^1$}
    child {
      node[label] {$G_3^2$}
      child {
        node[rectangle, draw, fill=black, minimum size=6pt] {}
        node[label, below=0.3cm] {$v_5$}
      }
      child {
        node[rectangle, draw, fill=black, minimum size=6pt] {}
        node[label, below=0.3cm] {$v_6$}
      }
    }
  };

\node[rectangle, draw, fill=black, minimum size=6pt] (v7) at ([yshift=-1.4cm, xshift=0.75cm]g21) {};
\node[label, below=0.3cm] at (v7) {$v_7$};

\node[rectangle, draw, fill=black, minimum size=6pt] (v8) at ([yshift=-1.4cm, xshift=1.25cm]g21) {};
\node[label, below=0.3cm] at (v8) {$v_8$};

\draw (g21) -- (v7);
\draw (g21) -- (v8);

\end{tikzpicture}
        }
    \end{minipage}
    \caption{
       This figure illustrates the hierarchical group structure built in Phase 1 of \NestedBased\ for eight individuals ($v_1,v_2$ at location 1; $v_3,v_4$ at 2; $v_5$–$v_8$ at 4) with parameters $\ell=2$ and $k=2$. The left-hand side shows two calls to \MGC: first with $t=2$, which forms $\G^2=\{G^2_1,G^2_2,G^2_3\}$ (leaving $v_7$ and $v_8$ ungrouped), and then with $t=1$, applied to these groups and the remaining individuals to obtain $\G^1=\{G^1_1,G^1_2\}$. The right-hand side depicts the resulting tree, with groups as internal nodes and individuals as leaves.
        }
    \label{fig:combined}
\end{figure}

\begin{lemma} \label{lem:hierarchical}
Let \( P \) be a panel containing a representative from each group \( G \in \G^j \). Then \( P \) satisfies \( O(4^{\ell - j}) \)-PFC for population \( V \) and panel size \( j \cdot k \).
\end{lemma}

\textbf{Panel Assignment.} The algorithm ensures representation for both individual panels and prefix panels by utilizing the constructed tree structure \( \Rr \) and the above lemma by enforcing the following two properties:
\begin{itemize}
    \item[(a)] Each individual panel contains one representative from each group \( G \in \G^1 \), ensuring approximate PFC representation for every panel.
    \item[(b)] For each group \( G \in \G^t \), there exists a representative from \( G \) among the first \( t \) panels, ensuring approximate PFC representation for every prefix up to time \( t \).
\end{itemize}
To establish these properties, we first introduce the \FR\ algorithm, which the algorithm employs in the second phase.
\FR\ takes as input the tree  $\Rr$ and a target group 
$G$, corresponding to a node in the tree.
Starting from  $G$, it recursively selects a subgroup contained in  $G$ that belongs to some $\G^j$  with the smallest possible index $j$, continuing this process until it reaches a group consisting only of leaf nodes.
Intuitively, \FR\ traces a path from the target group down to a leaf node which we call \emph{trajectory}, always moving through subgroups that are required to be represented earlier. Upon reaching the terminal node, it arbitrarily selects $\sfrac{n}{(\ell \cdot k)}$ leaves (equivalently saying individuals) as the sampling group $Q$. Notice that by construction any formed group either contains at least $\sfrac{n}{(\ell \cdot k)}$ individuals or another group,   implying that such a trajectory always exists.

\begin{algorithm}[H]
\caption{\FR}\label{alg:FR}

\begin{algorithmic}[1]
\STATE  \textbf{Input:}  Target group $G$
\STATE \textbf{Output:} Modified group $G'$ and sample group $Q$

\IF{$G$ contains a group $H \in \G^{i}$ for some $i \leq \ell$}
    \STATE Let $H^* \in G$ be the group in $\G^i$ with minimum index $i$
        \STATE  $ H', Q \leftarrow \FR( H^*)$
    \STATE Define modified target group $G' \leftarrow G \setminus \{H^*\} \cup H'$
    \STATE Return $G'$, $Q$
\ELSE
    \STATE Let $Q \subseteq G$ be an arbitrarily selected subset of size $\frac{n}{\ell \cdot k}$
    \STATE Return $G \setminus Q$, $Q$
\ENDIF

\end{algorithmic}
\end{algorithm}

Then, in two sub-phases, \NestedBased\ enforces the two properties respectively, as follows.

\paragraph{Satisfying Property (a) - Panel Representation:} This property requires that every group $G \in \G^1$ has a distinct representative in each of the $\ell$ panels. The algorithm addresses this requirement by processing groups individually through a sequential panel assignment procedure. For each group $G \in \G^1$, the algorithm iterates through panels $P_1, P_2, \ldots, P_\ell$ in order, seeking to assign one representative from $G$ to each panel. This systematic approach is feasible because group $G$ initially contains at least $\sfrac{n}{k}$ individuals, while each execution of \FR\ consumes only $\sfrac{n}{(\ell \cdot k)}$ individuals through the sampling group. Consequently, the algorithm can perform exactly $\ell$ iterations to populate all required panels.

However, a complication arises when $G$ contains subgroups that demand earlier representation. Specifically, when processing panel $P_t$, the algorithm first examines whether any subgroup $H \in G$  belongs to $\G^j$ for some $j < t$. Such a subgroup might violate the prefix representation guarantee when it is represented in panel $P_t$, since it requires representation in the first $j$ panels. To resolve this conflict, the algorithm extracts this problematic subgroup from $G$ and relocates it as direct children of the root $\Rr$, setting it aside for representation during the prefix sub-phase. Crucially, $H$ has not been represented yet and we are simply deferring its assignment to ensure it gets placed in an appropriately early panel.

The key is that removing $H$ and its entire sub-tree still preserves sufficient population in $G$ to support $\ell$ executions, as we establish in our analysis in~\Cref{sec:proof-of-exponential}.
After removing problematic subgroups, all remaining groups in $G$ lives in some $\G^i$ with $i \geq t$, allowing for the safe execution of \FR\ to obtain sampling group $Q$. The algorithm then samples an individual $v \in Q$, assigns $v$ to panel $P_t$ (where $v$ represents all groups along execution trajectory of \FR), and updates the tree structure.

The tree update process permanently removes all sampled individuals in $Q$ from future consideration, eliminates all intermediate groups along the trajectory except the root group $G$, and flattens the structure by having $G$ directly contain the remaining children groups of these intermediate groups. After completing all $\ell$ panels for group $G$, the algorithm removes $G$ from $\Rr$ and promotes any remaining children to become direct children of the root.

\paragraph{Satisfying Property (b) - Prefix Representation:} Groups requiring prefix representation (i.e., those set aside in the previous step) plus any other children of $\Rr$, accumulate as children of $\Rr$. While the root contains groups ($\Rr \neq \emptyset$), the algorithm executes \FR\ starting from $\Rr$ to obtain sampling group $Q$. It samples an individual $v \in Q$ and assigns $v$ to panel $P_t$ where $t$ is the minimum index such that $|P_t| < k$, and then it applies the same tree structure updates as before. As we show in the analysis,  if a group on execution trajectory of \FR belongs to $\G^t$ for some minimum value $t$, then there must be available capacity in the first $t$ panels to accommodate the selected representative. The algorithm terminates when $\Rr = \emptyset$, at which point all groups have received appropriate representation in accordance with their prefix requirements.

    \begin{algorithm}[!hbt]
    \caption{\NestedBased}\label{alg:TSA}

    \begin{algorithmic}[1]
   \STATE  \textbf{Input:} $V$, $d$, Groups $\G$, Parameter $\ell$, Panel size $k$.
    
    \STATE \textbf{Output:} Panels $P_1, \ldots, P_\ell$.
    \STATE \textbf{---Phase 1: Tree Formation---}
    \STATE Initialize tree $\Rr$ with children $V$, i.e., $\Rr \leftarrow V$
    \FOR{$t = \ell$ \textbf{down to} $1$}
        \STATE 
            Let $\G^t \leftarrow \MGC(V, d, t \cdot k, \Rr)$
        
        \FOR{$H \in \G^t$}
            \STATE Remove groups in $H$ from the root $\Rr$
            \STATE Add $H$ as a child of the root $\Rr$
        \ENDFOR
    \ENDFOR
    \STATE \textbf{---Phase 2: Panel Assignment---}
    \STATE \textit{---Phase 2.1: Panel Representation---}
    \FOR{$G \in \G^1$}
        \FOR{$t=1$ \textbf{to} $\ell$ } \label{line:ell-step-loop}
            \WHILE{there is a group $H \in G \cap \G^i$ for some $i < t$}  \label{line:ensures-excess-individuals}
                \STATE $G \leftarrow G \setminus \{H\}$ and $\Rr \leftarrow \Rr \cup \{H\}$.
            \ENDWHILE
                \STATE  $ G', Q \leftarrow \FR(G)$.
            \STATE Sample $v \in Q$ uniformly, assign to $P_t$
            \STATE Update $G \leftarrow G'$
        \ENDFOR
        \STATE Remove $G$ from $\Rr$, add its children to $\Rr$
    \ENDFOR
     \STATE \textit{---Phase 2.2: Prefix Representation---}
    \WHILE{$\Rr \neq \emptyset$}
        \STATE  $\Rr', Q \leftarrow \FR(\Rr)$
        \STATE Sample $v \in Q$, assign to $P_t$ with min $t$ s.t. $|P_t| < k$
        \STATE Update the root $\Rr \leftarrow \Rr'$
    \ENDWHILE
    \STATE \textbf{return} $P_1, \ldots, P_\ell$
    \end{algorithmic}
    \end{algorithm}

These two properties together with Lemma~\ref{lem:hierarchical} imply representation guarantees of Theorem~\ref{thm:exponential}.

\subsection{Modified Greedy Capture}
\label{sec:mgc}

In this section, we present the full algorithmic details of \MGC, a variation of \citet{chen2019proportionally}'s algorithm that serves as a key subroutine. \MGC\ takes input population \( V \), metric \( d \), target panel size \( K \), and a partition \( \G \) of disjoint groups over \( V \). The algorithm initializes an empty set $\G' = \emptyset$ to store the resulting groups.

The algorithm expands balls around each individual, initially marked incomplete. A ball captures a group in $\G$ when it encapsulates \emph{all} group members. Crucially, partial group capture does not count: when a ball captures some but not all members of a group, neither the partially captured individuals nor the group itself are considered captured by the ball. Only when every member of a group lies within the ball's radius is the entire group deemed captured. When an incomplete ball captures at least \( \sfrac{n}{K} \) individuals, it becomes \emph{complete}, its captured groups are consolidated as a single group and added to $\G'$, then disregarded from further processing.
The algorithm continues expanding all  balls, disregarding any groups in $\G$ that are captured by complete balls. New groups are formed when incomplete balls capture total of at least $\sfrac{n}{K}$ uncovered individuals. The process terminates when all groups are disregarded, returning $\G'$. 

Note that by construction, each group $G \in \G$ is either entirely contained within some group $G' \in \G'$ or entirely excluded from $\G'$. For each group \( G \) created during the execution of \MGC, we denote by \( c_G \) its \emph{center}, i.e., the point around which the ball was grown when the group was formed, and by \( r_G \) its \emph{radius}, defined as the radius of the ball at the moment \( G \) was created. 

For nested groups, we define the recursive function \(I(\cdot)\) returning all individuals in group \(G\) and its descendants. 
Here, groups are understood as sets that may contain either individuals or other groups, forming a hierarchical structure:
\[
I(G) = \begin{cases}
    \{G\} & \text{if } G \in V, \\
    \bigcup_{H \in G} I(H) & \text{otherwise}.
\end{cases}
\]

Notice that for a single individual $v \in V$, $I(v)$ returns a singleton set containing $v$, and for hierarchical groups, it returns the set containing all individuals by recursively traversing the group and all its descendants.

\begin{algorithm}[H]
\caption{\MGC}\label{alg:MGC}

\begin{algorithmic}[1]
\STATE \textbf{Input:}  $V$,  $d$, Parameter $K$, Groups $\G$
\STATE \textbf{Output:} Collection of groups $\G'$

\STATE Initialize $\delta \leftarrow 0$
\STATE Let $U \leftarrow \G$ be the collection of uncovered groups
\STATE Let $\G' \leftarrow \emptyset.$
\STATE Define $B^{\G}(x,\delta) := \{G \in U : I(G) \subseteq B(x, \delta)\}$
\WHILE{$U \neq \emptyset$}
    \STATE Increase $\delta$ continuously
    \FOR{each group $G \in \G'$}
        \STATE Let $S = B^{\G}(c_G, \delta)$ where $c_G$ is the center of $G$.
        \STATE Disregard groups in $S$ from the uncovered groups collection, i.e., $U \leftarrow U \setminus S$.
    \ENDFOR
    \WHILE{there exists $x \in V$ such that $\left|\bigcup_{H \in B^{\G}(x, \delta) \cap U} I(H) \right| \geq \frac{n}{K}$}
        \STATE Let $N = B^{\G}(x, \delta) \cap U$
        \STATE Capture groups in $N$, i.e., $\G' \leftarrow \G' \cup \{N\}$
        \STATE Disregard groups in $N$ from the uncovered groups collection, i.e., $U \leftarrow U \setminus N$
    \ENDWHILE
\ENDWHILE
\STATE \textbf{return} $\G'$
\end{algorithmic}
\end{algorithm}

\subsection{Proof of Lemma~\ref{lem:hierarchical}}
\begin{proof}[Proof of Lemma~\ref{lem:hierarchical}]
Fix an arbitrary index $j \in [\ell]$. Consider an arbitrary coalition $S$ with $|S| \geq \sfrac{n}{(j \cdot k)}$, and suppose there exists an individual $x \in S$ such that $d(x,y) \leq r$ for all $y \in S$ and some $r \geq 0$. We will prove that there exists a panel member $p \in P$ such that $\min_{y \in S} d(y, p) \leq r \cdot 2 \cdot 4^{\ell-j+1}$.

For each index $t \in \{j, j+1, \ldots, \ell\}$, let $\G^t$ denote the collection of groups formed during the execution of \MGC\ at iteration $t$ in the \NestedBased\ algorithm. For any individual $v \in S$ and index $t$, define $G_v^t \in \G^t$ as the unique group in $\G^t$ that disregards individual $v$ during the execution of \MGC. We say that individual $v$ is \emph{processed} by group $G_v^t$ if $v$ is captured while the ball centered at $c_{G_v^t}$ was still incomplete and included in $G_v^t$, or if $v$ is captured after the ball has opened and is simply disregarded.

We prove by backward induction on $t$ (from $t = \ell$ down to $t = j$) that for every individual $v \in S$, the following properties hold:
\begin{itemize}
    \item[(a)] \textbf{Proximity:} $d(c_{G_v^t}, v) \leq r \cdot (4^{\ell-t+1}-1)$
    \item[(b)] \textbf{Bounded radius:} $\max_{y \in I(G_v^t)} d(c_{G_v^t}, y) \leq r \cdot (4^{\ell-t+1}-1)$
\end{itemize}
where $c_G$ denotes the center of group $G$ and $I(G)$ denotes the set of individuals in group $G$.

\textbf{Base Case ($t = \ell$) :}
In this case, \MGC\ operates with parameter $K = \ell \cdot k$, so groups become complete when they contain at least $\sfrac{n}{(\ell \cdot k)}$ individuals. We claim that there exists an individual $v^* \in S$ such that both $d(c_{G_{v^*}^\ell}, v^*) \leq r$ and $\max_{y \in I(G_{v^*}^\ell)} d(c_{G_{v^*}^\ell}, y) \leq r$.

Suppose for contradiction that no such individual exists. Then when \MGC\ reaches radius $\delta = r$, no member of $S$ has been disregarded by any complete group. However, since $|S| \geq \sfrac{n}{(j \cdot k)} \geq \sfrac{n}{(\ell \cdot k)}$ and all members of $S$ lie within distance $r$ of point $x$, the ball $B(x, r)$ contains at least $\sfrac{n}{(\ell \cdot k)}$ individuals from $S$. By the algorithm's design, this ball would become complete and capture some of these individuals, disregarding all others in $S$, which contradicts our assumption. Therefore, such a $v^*$ exists.

For this $v^*$, properties (a) and (b) hold with bound $r = r \cdot (4^{\ell-\ell+1}-1) = r \cdot (4^1 - 1) \leq 3r$. For any other individual $v \in S$, the triangle inequality yields:
$$d(c_{G_{v^*}^\ell}, v) \leq d(c_{G_{v^*}^\ell}, v^*) + d(v^*, x) + d(x, v) \leq r + r + r = 3r.$$
Since every individual $v \in S$ must be processed by some group (either $G_{v^*}^\ell$ or another group with appropriately bounded center and radius), the base case is established.

\textbf{Inductive Step ($t < \ell$) :}
Assume the inductive hypothesis holds for all indices $t' > t$. We prove the claim for index $t$. At this step, \MGC\ operates with parameter $K = t \cdot k$, so groups become complete when they contain at least $\sfrac{n}{(t \cdot k)}$ individuals.

\begin{claim}
    There exists an individual $v^* \in S$ such that:
    \begin{itemize}
        \item $d(c_{G_{v^*}^t}, v^*) \leq r \cdot (1 + 2 \cdot (4^{\ell-t} - 1))$,
        \item $\max_{y \in I(G_{v^*}^t)} d(c_{G_{v^*}^t}, y) \leq r \cdot (1 + 2 \cdot (4^{\ell-t} - 1))$.
    \end{itemize}
\end{claim}
\begin{proof}[Proof of Claim]
    Suppose for a contradiction that no such individual exists. 
    By the inductive hypothesis, for each $v \in S$, we have: 
    $$d(c_{G_v^{t+1}}, v) \leq r \cdot (4^{\ell-t} - 1).$$
    During the execution of \MGC\ at index $t$, when $\delta = r \cdot (1 + 2 \cdot (4^{\ell-t} - 1))$, consider the ball $B(x, r)$. Recall that $S \subseteq B(x, r)$.
    
    Now note that for any $v \in S$ and any $y \in I(G_v^{t+1})$:
    \begin{align*}
    d(x, y) &\leq d(x, v) + d(v, c_{G_v^{t+1}}) + d(c_{G_v^{t+1}}, y)\\
    &\leq r + r \cdot (4^{\ell-t} - 1) + r \cdot (4^{\ell-t} - 1) \\
    &= r \cdot (1 + 2 \cdot (4^{\ell-t} - 1)).
    \end{align*}
    
    Therefore, the ball $B(x, r \cdot (1 + 2 \cdot (4^{\ell-t} - 1)))$ contains all individuals from groups $\{G_v^{t+1} : v \in S\}$, which includes at least $|S| \geq \sfrac{n}{(j \cdot k)} \geq \sfrac{n}{(t \cdot k)}$ individuals. By the algorithm's design, this ball becomes complete, contradicting our assumption. Therefore, such a $v^*$ exists. 
\end{proof}
For all other individuals $v \in S$, the triangle inequality ensures that properties (a) and (b) hold with the bound:
$$r \cdot (1 + 2 \cdot (4^{\ell-t} - 1) + 2 + 2 \cdot (4^{\ell-t} - 1)) = r \cdot (4^{\ell-t+1} - 1).$$
Applying the inductive claim at index $j$, we obtain that there exists a group $G \in \G^j$ such that:
$$\min_{v \in S} d(c_G, v) \leq r \cdot (4^{\ell-j+1} - 1) \leq r \cdot 4^{\ell-j+1}$$
Since panel $P$ contains a representative from each group in $\G^j$, there exists $p \in P$ with $d(p, c_G) \leq \max_{y \in I(G)} d(c_G, y) \leq r \cdot 4^{\ell-j+1}$.
By the triangle inequality:
$$\min_{v \in S} d(v, p) \leq \min_{v \in S} d(v, c_G) + d(c_G, p) \leq r \cdot 4^{\ell-j+1} + r \cdot 4^{\ell-j+1} = r \cdot 2 \cdot 4^{\ell-j+1}$$
This establishes $O(4^{\ell-j})$-PFC.
\end{proof}

\subsection{Proof of Theorem~\ref{thm:exponential}}
\label{sec:proof-of-exponential}

We start by proving some auxiliary lemmas first. 

\begin{lemma}
    \label{lem:find-repr-terminates}
   \FR\ algorithm always terminates properly.
\end{lemma}

\begin{proof}
    We prove \FR\ terminates in both Phase~2.1 and Phase~2.2 of the \NestedBased\ algorithm.
    
    During Phase~2.1, the algorithm processes each group $G \in \G^1$ by executing \FR\ exactly $\ell$ times (once for each panel). Each execution removes groups in the execution trajectory and a sampling group of size $\sfrac{n}{(\ell \cdot k)}$ from $G$. Importantly, while the contents of $G$ change, the groups in $\bigcup_{i=2}^\ell \G^i$ maintain their structure through out the process, as they are only removed from the hierarchy, not modified internally.
    
    First, we establish a baseline: when \FR\ is executed on any group whose content has not been modified, it must terminate. This follows because every recursive path eventually reaches a terminal group containing at least $\sfrac{n}{(\ell \cdot k)}$ individuals, which is sufficient to create the required sampling group.
    
    Now suppose, for contradiction, that \FR\ fails to terminate when executed on group $G$ during Phase~2.1. Since \FR\ would succeed on any unmodified subgroup (by our baseline observation), the failure can only occur if $G$ contains only individuals (no subgroups) and $|I(G)| < \sfrac{n}{(\ell \cdot k)}$.
    
    Let $t^*$ be the first iteration where, after executing \FR, we have $|I(G_{t^*})| < (\ell-t^*) \cdot \sfrac{n}{(\ell \cdot k)}$. Here, $G_{t^*}$ denotes the content of $G$ after completing iteration $t^*$. Such a $t^*$ must exist since we assumed $|I(G)| < \sfrac{n}{(\ell \cdot k)}$ eventually. By the minimality of $t^*$, after iteration $t^*-1$ we have:
    $$|I(G_{t^*-1})| \geq (\ell-t^*+1) \frac{n}{\ell \cdot k}.$$
    
    Since each execution of \FR\ removes exactly $\sfrac{n}{(\ell \cdot k)}$ individuals from $G$, and we know that $|I(G_{t^*-1})| - \sfrac{n}{(\ell \cdot k)} < (\ell-t^*) \sfrac{n}{(\ell \cdot k)}$, the deficit at iteration $t^*$ must trigger the while-loop condition at line~\ref{line:ensures-excess-individuals}. This while-loop removes groups from $G$ that belong to $\G^i$ for $i < t^*$. Let $H_{t^*}$ be one such group removed during this while-loop.
    
    Let us denote by $G'_{t^*}$ the content of $G$ at iteration $t^*$ right after the while-loop completes but before calling \FR. This is the critical state we need to analyze.
    
    Let $L_1, \ldots, L_{t^*-1}$ be the collection of trajectories followed while executing \FR in iterations $1$ through $t^*-1$ respectively, with corresponding sampling groups $Q_1, \ldots, Q_{t^*-1}$. For each trajectory $L_i$, let $H_i$ be the group on $L_i$ that belongs to $\G^{t_i}$ where $t_i$ is maximal subject to $t_i \leq t^*-1$.
    
    We establish two crucial properties:
    \begin{enumerate}
        \item For any $i \in [t^*-1]$, the group $H_i$ cannot contain any group from $\G^t$ for $t_i \leq t \leq t^*-1$. If it did, \FR\ would have recursed through that group, contradicting the maximality of $t_i$ for $H_i$.
        
        \item The groups $H_1, \ldots, H_{t^*-1}, H_{t^*}$ are mutually disjoint. To see why, note that if $I(H) \cap I(H') \neq \emptyset$ for two groups in our collection, then by the hierarchical structure, either $H \subseteq H'$ or $H' \subseteq H$. But property (1) prevents any $H_i$ from containing another $H_j$ when $i, j \leq t^*-1$. In addition, $H_{t^*} \in \G^i$ for $i < t^*$ ensures it cannot be contained in any $H_j$ and $H_{t^*}$ cannot contain any $H_j$'s as they appear in some other trajectories implying that $H_{t^*}$ cannot be their ancestor in the hierarchy.
    \end{enumerate}

    Property (1) has an important consequence: the individuals in $\bigcup_{i=1}^{t^*-1} (I(H_i) \setminus Q_i)$ cannot be removed during the while-loop at iteration $t^*$. This is because these individuals either belong to groups in $\G^t$ for $t \geq t^*$ (and thus don't satisfy the while-loop condition) or are direct members of $G$ not belonging to any subgroup.
    
    Now we can lower bound $|I(G'_{t^*})|$. Using property (2) and the fact that each $H_i$ contains at least $\frac{n}{(t^*-1) \cdot k}$ individuals:
    \begin{align*}
    |I(G'_{t^*})| &\geq \left|\bigcup_{i=1}^{t^*-1} (I(H_i) \setminus Q_i)\right|\\
    &= \sum_{i=1}^{t^*-1} |I(H_i) \setminus Q_i| \quad \text{(by disjointness of the $H_i$'s)}\\
    &= \sum_{i=1}^{t^*-1} (|I(H_i)| - |Q_i|)\\
    &\geq (t^*-1) \cdot \frac{n}{(t^*-1) \cdot k} - (t^*-1) \cdot \frac{n}{\ell \cdot k}\\
    &= \frac{n}{k} - \frac{(t^*-1) \cdot n}{\ell \cdot k}\\
    &= (\ell - t^* + 1) \cdot \sfrac{n}{(\ell \cdot k)}
    \end{align*}
    But this means $G'_{t^*}$ contains at least $(\ell - t^* + 1) \cdot \sfrac{n}{(\ell \cdot k)}$ individuals. Therefore, \FR\ at this step can terminate properly and removes $\sfrac{n}{(\ell \cdot k)}$ individuals implying that $|I(G_{t^*})| \geq (\ell - t^*)\frac{n}{\ell \cdot k}$, contradicting our assumption.
    
    During Phase~2.2, a similar counting argument shows that the root $\Rr$ maintains $|I(\Rr)| \geq \sfrac{n}{(\ell \cdot k)}$ until it becomes empty. Moreover, whenever $\Rr$ contains groups (not just individuals), \FR\ can recurse through these groups and terminate successfully. This completes the proof.
\end{proof}

\begin{corollary}
    \label{cor:panel-level-repr}
    For any panel $P_i$ and any group $G \in \G^1$, we have $P_i \cap I(G) \neq \emptyset$.
\end{corollary}
\begin{proof}
    This corollary follows from Lemma~\ref{lem:find-repr-terminates}, as \FR\ is executed $\ell$ times for each group $G \in \G^1$.
\end{proof}

\begin{lemma}
    \label{lem:prefix-level-repr}
    For any cumulative panel $P_{\leq i}$ and any group $G \in \G^i$, we have $P_{\leq i} \cap I(G) \neq \emptyset$.
\end{lemma}
\begin{proof}
    For any index $t$, if a group $G \in \G^t$ is removed from the tree during the second phase of \NestedBased, by construction an individual $v \in I(G)$ is assigned to one of the first $t$ panels.

    Assume for contradiction that there exists a group $G \in \G^t$ not represented by the first $t$ panels, and assume $t$ is the minimum such index. When $G$ is removed from the tree, $G$ must belong to $\Rr$ and \FR\ will provide a sample group belonging to $I(G)$. \NestedBased\ can only violate the representation guarantee if the first $t$ panels are completed.

    Note that at the end of Phase~2.1 (Panel Representation), each panel contains exactly $|\G^1|$ individuals. Moreover, since the first $t$ panels are full, there must be at least $t \cdot (k-|\G^1|)+1$ groups belonging to $\G^t$ existing in the tree at the beginning of Phase~2.2 (Prefix Representation). Since $\G^t$ groups are disjoint and sample groups obtained by running \FR\ during Phase~2.1 are disjoint from these $\G^t$ groups, we obtain
    $$|V| \geq \frac{n}{\ell \cdot k} \cdot \ell \cdot |\G^1| + (t \cdot (k-|\G^1|+1)) \cdot \frac{n}{t \cdot k} \geq n + \frac{n}{t \cdot k} > n.$$
    This yields a contradiction.
\end{proof}

\begin{lemma}
    \label{lem:fairness}
    The output panels of \NestedBased\ satisfy individual fairness.
\end{lemma}
\begin{proof}
    \FR\ is executed a total of $\ell \cdot k$ times, outputting a sample group $Q$ of size $\frac{n}{\ell \cdot k}$ after each execution. Sample groups from these executions form a partition of $V$. Since the algorithm selects uniformly random individuals from each sample group, it ensures every individual appears in some panel with probability $\frac{\ell \cdot k}{n}$.
\end{proof}

\begin{proof}[Proof of \cref{thm:exponential}]
    Panel level representation guarantee follows from \Cref{lem:hierarchical} and \Cref{cor:panel-level-repr}. Prefix level representation follows from \Cref{lem:hierarchical} and \Cref{lem:prefix-level-repr}. Individual fairness follows from Lemma~\ref{lem:fairness}.
\end{proof}

\section{Prefix Level Representation}\label{sec:prefix}

In the previous section, we showed that by nesting groups across panel sizes, we can ensure approximate PFC representation for both prefixes and individual panels. However,  the approximation factor increases exponentially with the number of panels $\ell$.

In this section, we adopt a different approach and show that a \emph{constant}-factor 
approximation to PFC can be achieved with respect to every prefix panel 
\( P_{\leq t} \), though at the cost of abandoning representation guarantees for 
individual panels. The key idea is to replace the nested-group construction with 
a distinct family of groups \( \mathcal{G}^t \) for each prefix \( t \in [\ell] \), 
capturing the groups that must be represented within the first \( t \) panels. 
Unlike the nested setting, this method does not require that for overlapping groups 
$G \in \mathcal{G}^t$ and $G' \in \mathcal{G}^{t'}$ with $t < t'$, the group from the 
larger prefix size ($G'$) be fully contained in the group from the smaller prefix size 
($G$), i.e., $G' \subseteq G$. Instead, the algorithm links groups across 
different prefix sizes into \emph{chains}, i.e.,  sequences of the form $ G^t \rightarrow G^{t+1} \rightarrow \cdots \rightarrow G^{\ell}$, with $G^j \in \G^j$,
where each group overlaps with some of its predecessor and has radius no larger than those that overlaps with. Representing the final group in such a chain then suffices to approximate the 
representation of all groups in the sequence, ensuring that coverage for the last 
group automatically yields coverage for every earlier one.

Note that simply linking overlapping groups with radii no larger than their predecessors is not enough, as the approximation error can accumulate at each step, leading to a bound that grows linearly with~$\ell$. To prevent this, we build the chains in a  more careful way.

\begin{theorem}\label{thm:constant}
    There exists a polynomial time algorithm that returns a distribution over a sequence of $\ell$ panels $P_1, \ldots, P_{\ell}$, where panel $P_t$ has size $k_t$,  such that:
    \begin{itemize}
         \item For each $t\in [\ell]$,  prefix panel $P_{\leq t}=\cup_{j=1}^t P_j$  satisfies $O(1)$-PFC;
         \item Each individual is selected in $\cup_{t\in [\ell]} P_{t}$ with probability $ \sfrac{ \sum_{t=1}^{\ell}k_t}{n}$. 
    \end{itemize}
     
\end{theorem}

\noindent
\paragraph{Description of the Algorithm.}
Our algorithm, called \ChainBased\ consists of three main phases.

In the first phase, for each \( t \in [\ell] \), the algorithm runs \MGC\ on the population 
\( V \) with metric \( d \), panel size \( \sum_{j=1}^t k_j \), and with each individual 
initially assigned to their own singleton group (i.e., there is no pre-existing grouping 
that must be nested). This yields a collection of groups \( \mathcal{G}^t \), representing 
the  groups that should be represented among the first \( t \) panels.

In the subsequent phases, the algorithm aims to build panels $P_1,\ldots,P_\ell$ so that, for every prefix $P_{\leq t}$ and every group $G^t \in \G^t$, there exists an individual in $P_{\leq t}$ whose distance to $c_{G^t}$ is at most $\alpha r_{G^t}$ for a universal constant $\alpha \ge 1$. We say that a group is \emph{covered} once this condition holds, and the goal is to cover all groups in $\bigcup_{t\in[\ell]}\G^t$. 
To achieve this, the algorithm relies on the following key property of the groups generated by \MGC, which is formalized in the following lemma. 

\begin{lemma}
    \label{lem:intersect-close}
    Let $V$ be a population in a metric space with distance function $d$, and consider panel sizes $k_1 \leq k_2$. Let $\G$ and $\H$ be the outputs of \MGC\ with panel sizes $k_1$ and $k_2$, respectively. For every group $G \in \G$ with center $c_G$ and radius $r_G$, there exists a group $H \in \H$ with center $c_H$ and radius $r_H$ such that $d(c_H, c_G) \leq 2 \cdot r_G$ and $r_H \leq r_G$.
\end{lemma}
\begin{proof}
Consider any group $G \in \G$ formed during the execution of \MGC\ with panel size $k_1$. When group $G$ is created, the radius threshold has reached $r_G$, and $G$ contains at least $\sfrac{n}{k_1}$ individuals that were previously not disregarded. Since $k_1 \leq k_2$, we have $\sfrac{n}{k_1} \geq \sfrac{n}{k_2}$, which implies $|G| \geq \sfrac{n}{k_2}$.

Consider the execution of \MGC\ with panel size $k_2$. Let $\delta$ denote the radius threshold at any point during this execution. We claim that when $\delta = r_G$, at least one individual in $G$ must be disregarded by some group in $\H$. If all individuals in $G$ remain available at this point, then the ball centered at $c_G$ would capture all $|G| \geq \frac{n}{k_2}$ individuals and become complete, forming a new group in $\H$.

Let $v \in G$ be an individual that is disregarded by some group $H \in \H$ when $\delta = r_G$. Since $H$ captures $v$ at radius threshold $r_G$, we have $r_H \leq r_G$. Moreover, by the triangle inequality:
$$d(c_H, c_G) \leq d(c_H, v) + d(v, c_G) \leq r_H + r_G \leq r_G + r_G = 2 \cdot r_G.$$
This proves the lemma.
\end{proof}

The above lemma indicates that for every group $G^t \in \G^t$  and every  $j\geq t$, there exists a group $G^{j} \in \G^{j}$ such that (i)  $G^t$ and $G^{j}$ overlap and 
 (ii) the  radius of $G^{j}$ is at most equal to the radius of $G^t$.  Essentially this means that by (approximately) representing group $G^{j}$, then $G^t$ is also (approximately) represented.

By exploiting this lemma, the second phase proceeds as follows. Initially,  all the groups are marked as uncovered.    For each group $G^t$, the algorithm attempts to find a group in $\mathcal{G}^\ell$ such that once this group is represented, $G^t$ is also approximately represented. More precisely, the algorithm starts from the set $\mathcal{G}^t$ with uncovered groups of smallest index and picks an arbitrary uncovered group $G^t$. It then calls the subroutine \findrepr, which tries to construct a chain of groups $G^t\rightarrow \ldots\rightarrow G^\ell$, with each $G^j \in \mathcal{G}^j$, ensuring that selecting an individual from $G^\ell$ represents every group in the chain within a constant approximation. The subroutine begins by marking $G^t$ as the \emph{anchor} and uses it to guide the chain extension. For each $j \in \{t+1,\ldots,\ell\}$, it finds a group $G^j \in \mathcal{G}^j$ that satisfies the conditions of Lemma~\ref{lem:intersect-close} relative to the anchor. If such a group is found and is uncovered, it is added to the chain; furthermore, if the radius of $G^j$ is at most half that of the current anchor, $G^j$ is promoted as the new anchor. This process continues until either a group in $\mathcal{G}^\ell$ is reached, at which point the chain construction succeeds, or a covered group is encountered, in which case \findrepr\ returns an unsuccessful chain consisting only of $G^t$. \ChainBased\  marks all groups in the returned chain as covered and repeats the process, always by calling \findrepr\ on an uncovered group in $\G^j$ with the smallest index $j$. The phase ends once all groups are marked as covered.

\begin{algorithm}[!hbt]
    \label{alg:temporal_sortition}
    \caption{\ChainBased}
        
    \begin{algorithmic}[1]\label{alg:constant}
      \STATE     \textbf{Input:}  $V$,  $d$, Panel sizes $k_1, \ldots, k_{\ell}$
    
     \STATE   \textbf{Output:}  A distribution over a sequence of panels $P_1,\ldots, P_{\ell}$ 
   \STATE \textbf{--- Phase 1: Construct Groups  ---}
        \FOR{level $t = 1$ to $\ell$}
            \STATE  $\G^t\leftarrow$  \MGC$( V, d, \sum_{j=1}^t k_j, \union_{v\in V}\{v\})$
        \ENDFOR
        
        \STATE \textbf{ ---Phase 2: Build Chains and Assign Priorities ---}
        \FOR{$t=1$ to $\ell-1$ } 
            \FOR{ each uncovered $G^t\in \G^t$}
                \STATE $\status$, $T$ $\leftarrow$ \findrepr($G^t$, $\{\G^1, \ldots, \G^\ell\}, V, d$ )
                \STATE Mark all groups in chain $T$ as covered
                    \STATE If status is succeed, assign to the last group in $T$ priority $t$
            
            \ENDFOR
        \ENDFOR

        \STATE \textbf{---Phase 3:  Sampling and Assignment to Panels---}
        \FOR{each $G^\ell \in \G^\ell$}
            \STATE Sample  $v$ from $G^\ell$ uniformly at random
            \STATE Assign $v$ the priority label of its group $G^\ell$
        \ENDFOR
    \STATE Sample $\sum_{t=1}^{\ell} k_t - |\G^\ell|$ representatives uniformly at random from $V \setminus \bigcup_{G \in \G^\ell} G$, and assign each a priority label of $\ell$. \label{alg-line-sampling}
    \STATE Assign representatives to panels by ordering them in increasing order of priority labels, and placing each into the earliest panel $P_t$ with available capacity (i.e., the smallest such index $t$)
        \RETURN $P_1, \ldots, P_\ell$
    \end{algorithmic}
\end{algorithm}

In the third phase, the algorithm uses the previously constructed chains to generate a distribution over panels \( P_1, \ldots, P_\ell \). To ensure individual fairness, it begins by sampling \( \sum_{t=1}^\ell k_t \) individuals from the population as follows: one individual is selected uniformly at random from each group in \( \mathcal{G}^\ell \), and the remaining \( \sum_{t=1}^\ell k_t  - |\mathcal{G}^\ell| \) individuals are sampled uniformly at random from the rest of the population\footnote{When $n$ is divisible by $\sum_{t=1}^\ell k_t$, this can be achieved by applying standard dependent rounding techniques. Otherwise, we first create groups of size exactly $\,\sfrac{n}{\sum_{t=1}^\ell k_t}$ by fractional assignment, and then apply Birkhoff's decomposition, as in~\cite{Ebadian2025}.} Next, to ensure proportional representation, the algorithm assigns each sampled individual a priority label based on the chains created in the previous phase. In particular, if an individual is sampled from a group $G^\ell$ that belongs to some chain $G^t\rightarrow \ldots \rightarrow G^\ell$, the individual is assigned the index $t$, corresponding to the head of the chain. If the individual is not sampled from any such group, they are assigned the priority label $\ell$.
Finally, the algorithm allocates individuals to panels sequentially, from $P_1$ to $P_\ell$, in increasing order of their assigned priority labels. Intuitively, this gives higher priority to individuals representing groups that must be covered earlier, ensuring that such representatives appear sooner in the panel sequence.

\begin{algorithm}[!hbt]
    \caption{\findrepr}\label{alg:find_representative}

    \begin{algorithmic}[1]

 \STATE      \textbf{Input:} Target group $G^t \in \G^t$,  $\{\G^1, \ldots, \G^\ell\}$, $V$, $d$

  \STATE   \textbf{Output:} status (succeed when a new chain was constructed and failed otherwise) and chain $T$
       
        \STATE Initialize the chain and anchor: $T \leftarrow G^t$ and $H \leftarrow G^t$

        \FOR{ $j = t+1$ to $\ell$}
            \STATE 
                Let $G^j \in \G^j$ be such that $d(c_{H}, c_{G^j})\leq 2\cdot r_H$ and $r_{G^j}\leq r_H$        
            \IF{$G^j$ is already covered}
                \STATE Remove from $T$ any group except for $G^t$ 
                \RETURN $fail, T$
            \ENDIF

            \STATE Add to $G^j$ at the end of chain $T$
            
            \IF{${r_{G^j}} < r_H/2$ \textbf{(radius shrinks significantly)}}
                \STATE Update  anchor: $H \leftarrow G^j$
            \ENDIF
        \ENDFOR
        
        \RETURN $succeed, T$
    \end{algorithmic}
\end{algorithm}

Before proving \Cref{thm:constant}, we first prove two more auxiliary lemmas.

\begin{lemma}
    \label{lem:chain-pr-guarantee}
    For every group $G^t \in \G^t$, there exists a representative $v$ among the first $t$ panels, such that $d(c_{G^t}, v) \leq 16 \cdot r_{G^t}$.   
\end{lemma}
\begin{proof}
We start by showing that  whenever a group \( G^\ell \in \G^\ell \) is assigned a priority number \( i \), some representative from this group is included in one of the first \( i \) panels. To establish this, we show that for each \( t \), at most \( \sum_{j=1}^t k_j \) groups can have priority number at most \( t \). This implies that for every \( t \), there are at most \( \sum_{j=1}^t k_j \) representatives that must be assigned to the first \( t \) panels—exactly matching the number of available seats across these panels. Hence, in the final phase, the algorithm ensures that the representative from every group in $\G^{\ell}$ is assigned to a panel no later than its priority number.
We prove this by induction on \( t \).
For \( t = 1 \), there are at most \( |\G^1| \leq k_1 \) groups with priority number 1, satisfying the bound. Suppose the claim holds for \( t - 1 \). Then, at step \( t \), there are at most \( |\G^t| \leq k_t \) additional groups that can be assigned priority number \( t \). By the inductive hypothesis, at most \( \sum_{j=1}^{t-1} k_j \) groups have priority number at most \( t - 1 \). Therefore, the total number of groups with priority number at most \( t \) is at most
$\sum_{j=1}^{t-1} k_j + k_t = \sum_{j=1}^t k_j$,
as required.

   It remains to  show that for each $G^t \in \G^t$, with $t\in [\ell]$, there exists a representative that is at most $16 \cdot r_{G^t}$ far away from its center and is assigned a priority label at most $t$.  Note that,  during the second phase, the algorithm marks $G^t$ as covered after either including it in a chain of groups $G^j\rightarrow\ldots \rightarrow G^t \rightarrow  \ldots \rightarrow G^\ell$ or by marked it as covered when $G^t$ is the head of the chain and the algorithm fails to construct a new chain. We analyze the two cases separately.

    \paragraph{Case 1: $G^t$ is included in a chain.}
 
    When $G^t$ is first added to a chain with current anchor $H$, the algorithm checks whether $r_{G^t} \leq \sfrac{r_H}{2}$. If so, $G^t$ becomes the new anchor. Otherwise, the anchor remains $H$. In either case,  we denote the resulting anchor by $H_1$. The algorithm continues to add groups to the chain by performing a shrinking test on the radius. Each time the anchor changes, the radius decreases by at least a factor of $2$. Let 
\(H_1, H_2, \ldots, H_m \) denote the sequence of successive anchors during the execution of the algorithm after adding $G^t$ to the chain. Then, we have
\begin{align}
        d(c_{G^t}, c_{G^\ell}) 
            &\leq d(c_{G^t}, c_{H_1}) + \sum_{i=1}^{\max\{1, m-1\}} d(c_{H_{i}}, c_{H_{i+1}}) + d(c_{H_{m}}, c_{\ell})  \nonumber \\
            &\leq 2\cdot r_{G^t} +  \left( \sum_{i=1}^{\max\{1, m-1\}} 2\cdot r_{H_{i}} \right)  + 2\cdot r_{H_m} \nonumber
            \\
            &\leq 2\cdot r_{G^t} + 2\cdot r_{H_{1}} \cdot \left( \sum_{i=2}^{\max\{1, m-1\}} \frac{1}{2^i}  \right) + 2 \cdot r_{H_1} \nonumber \\
            &\leq 2\cdot r_{G^t} +2\cdot r_{G^t}  \cdot  2  +  2 \cdot r_{G^t}   \leq    8\cdot r_{G^i},  
    \end{align}
where  the first inequality follows by the triangle inequality, the second inequality follows by~\Cref{lem:intersect-close}, the third inequality follows since the radius shrinks by at least a factor of $2$  each time the anchor changes and the fourth inequality follows again by~\Cref{lem:intersect-close}. Moreover, if  $G^\ell$  is the end point of the chain, then $G^\ell$ is assigned a priority number that is at most $t$, as the head of the chain belongs to $\G^j$ with $j\leq t$. This means that when a representative from $G^\ell$ is chosen is assigned a priority number of at most $t$.

    \paragraph{Case 2: $G^t$ is not included in a chain.}
   
    When the algorithm fails to construct the chain, this means that during the process, the algorithm failed to find a group from some set $\G^j$ with $j>t$ that is uncovered and satisfies the above conditions.  This means that there exists a group $G^j \in \G^j$ that has been assigned to a previous chain with end point some $G^\ell\in G^\ell$ and from the above  case we know that    $d(c_{G^j}, c_{G^\ell}) \leq 8 \cdot r_{G^j}$. Moreover with similar arguments as above, we can conclude that $d(c_{G^t}, c_{G^j}) \leq 8 \cdot r_{G^t}$ and $r_{G^j} \leq r_{G^t}$ implying that $d(c_{G^t}, c_{G^\ell}) \leq 16 \cdot r_{G^t}$. Lastly, $G^\ell$ is assigned a priority number at most $t$, because the algorithm 
initiates chains from groups in $\G^r$ with the smallest index $r$ that still contain 
uncovered groups, and at the time this chain was started, $G^t$ was still uncovered.
\end{proof}

\begin{lemma}
    \label{lem:gc-pfc-guarantee}
        Let $V$ be a population and $k$ be a panel size. Consider the collection of groups $\G$ formed by $\MGC$ over $V$ with panel size $k$ and each individual initially assigned to their own group. If $P\subseteq V$ is a panel  such that, for every  $G\in \G$, there exists $p\in P$ with $d(p,c_G)\le \alpha\cdot  r_G$, then $P$ satisfies $(\alpha+3)$-PFC.
\end{lemma}

\begin{proof}
    Let $S \subseteq V$ be an arbitrary set of individuals with $|S| \geq \sfrac{n}{k}$. Let $y^* \in \argmin_{y \in V} \max_{u \in S} d(u,y)$ be the point that minimizes the maximum distance to any point in $S$, and define $r = \max_{u \in S} d(u, y^*)$ as this optimal radius. Let $v^* \in S$ be the farthest point from $y^*$, so $d(y^*, v^*) = r$.
    
     We first claim there exists a group $G \in \G$ with $r_G \leq r$ and some $v \in S$ satisfying $d(v, c_G) \leq r_G$. Indeed, if no such group existed, then when $\MGC$ reaches radius $r$, all points in $S$ would not be disregarded yet. Since $|S| \geq \sfrac{n}{k}$, the greedy algorithm would form a group containing some point from $S$ at a radius of at most $r$, leading to a contradiction.
    
    For this group $G$ with witness $v \in S$, we have
    $$d(y^*, c_G) \leq d(y^*, v) + d(v, c_G) \leq r + r_G \leq 2r,$$
    since $d(y^*, v) \leq r$ by the definition of $y^*$. Thus
    $$d(v^*, c_G) \leq d(v^*, y^*) + d(y^*, c_G) \leq r + 2\cdot r = 3 \cdot r.$$
    
    By assumption, there exists $p \in P$ with $d(c_G, p) \leq \alpha \cdot r_G \leq \alpha \cdot r$. Therefore,
    $$d(v^*, p) \leq d(v^*, c_G) + d(c_G, p) \leq 3\cdot r + \alpha \cdot r = (\alpha + 3)\cdot r.$$
    Since $r = \min_{y \in V} \max_{u \in S} d(u,y)$, we have shown that for $v^* \in S$ and $p \in P$, we have $d(v^*, p) \leq (\alpha + 3) \cdot r$. As $S$ was arbitrary, $P$ satisfies $(\alpha+3)$-PFC.
\end{proof}

\begin{proof}[Proof of~\Cref{thm:constant}]
    We establish both properties of the theorem. For the first property, we see that each prefix panel $P_{\leq t}$ satisfies $19$-PFC by  combining~\Cref{lem:chain-pr-guarantee} and~\Cref{lem:gc-pfc-guarantee}.

    For the second property, we analyze the selection probability for each individual. If an individual appears in a group $G \in \G^\ell$, then it appears in some panel with probability 
    $\sfrac{\sum_{t=1}^{\ell}k_t}{n}$
     since each such group  $\G^\ell$ contains exactly $\sfrac{n}{\sum_{t=1}^{\ell}k_t}$ individuals, and one of them is chosen uniformly at random.  For individuals not covered by groups in $\G^\ell$, the number of such individuals equals $\sfrac{n}{\sum_{t=1}^{\ell}k_t}$ times the number of empty panel seats,  and they are sampled with equal probability to fill the remaining positions. Hence, each individual appears in some panel with probability $\sfrac{\sum_{t=1}^{\ell}k_t}{n}$.\qedhere
    
\end{proof}

It remains an open question whether one can design a distribution over a sequence of panels such that every prefix satisfies a constant-factor approximation of the proportionality notion PRF, which, in contrast to 
PFC, requires that each subset of individuals receive not just one representative, but a number of representatives proportional to its size.

\begin{openq}
    Does there exist a  distribution over a sequence of panels, $P_1, \ldots, P_{\ell}$ such that each prefix panel $P_{\leq t}$  satisfies $O(1)$-PRF? 
\end{openq}

\section{Conclusion}
Permanent citizens' assemblies represent a promising democratic innovation by ensuring that minority voices can persist over time, even when too small to warrant representation in any single panel. We formalize this challenge as a problem of \emph{temporal sortition} in metric spaces and make progress on it by proposing three algorithms, each providing different guarantees  for panel-level  and  prefix-level  representation.

Beyond the questions highlighted above, several intriguing directions remain open. For instance, can we ensure representation not just for prefix panels, but for \emph{any} consecutive subsequence of panels? Additionally, is it possible to achieve meaningful guarantees without knowing in advance the total number or sizes of panels?

\bibliographystyle{plainnat}
\bibliography{abb,sample}
\end{document}